\documentclass[aps,pra,preprint,notitlepage,nofootinbib]{revtex4-2}

\usepackage{amsmath,amssymb,amsthm,mathtools,bm}
\usepackage{hyperref}
\usepackage{physics}
\usepackage{graphicx}
\usepackage{enumitem}
\usepackage{placeins}
\usepackage[T1]{fontenc}
\usepackage{lmodern}

\usepackage{xcolor}
\definecolor{linkblack}{RGB}{20,20,20}
\definecolor{citeblue}{RGB}{0,70,140}
\definecolor{urlblue}{RGB}{0,90,120}
\usepackage{float}
\hypersetup{
  colorlinks=true,
  linkcolor=linkblack,
  citecolor=citeblue,
  urlcolor=urlblue
}

\newtheorem{theorem}{Theorem}
\newtheorem{proposition}{Proposition}
\newtheorem{corollary}{Corollary}
\newtheorem{remark}{Remark}

\begin{document}


\title{Resonant delay in a stationary quantum clock:\\
Lifting the threshold mask}

\author{Paul~C.~W.~Davies} 
\email{paul.davies@asu.edu}
\affiliation{Department of Physics, Arizona State University, Tempe, Arizona 85287, USA}
\affiliation{Beyond Center for Fundamental Concepts in Science,  
Arizona State University, Tempe, AZ 85287, USA}

\author{Damien A. Easson}
\email{easson@asu.edu}
\affiliation{Department of Physics, Arizona State University, Tempe, Arizona 85287, USA}
\affiliation{Beyond Center for Fundamental Concepts in Science, 
Arizona State University, Tempe, AZ 85287, USA}


\begin{abstract}
Quantum transit times have a long history of inequivalent definitions, including
phase times, dwell times, and quantum-clock constructions.  In this context we
revisit the Salecker--Wigner--Peres stationary quantum clock as a phase-sensitive
scattering observable, with clock time defined by the energy derivative of the
transmission phase shift across the interaction region.  For real compactly
supported one-dimensional potentials, we show that the raw stationary Peres
clock generically contains a universal \(1/\sqrt{E}\) continuum-edge term whose
coefficient is fixed by low-energy scattering data.  For the attractive square
well, this threshold singularity is inherited from the vanishing exterior
momentum and the associated scattering matching, rather than from resonant
delay itself.
We derive the exact stationary clock time for the square well and introduce a
new threshold-subtracted clock observable.  Away from exceptional zero-energy
tuning, the subtraction removes the universal low-energy term and isolates the
resonant contribution.  Comparison with the dwell time and the transmission
Wigner phase delay shows that the threshold-subtracted clock acquires the
expected local Lorentzian form near isolated transmission resonances.  Near the
continuum edge, if \(\varepsilon\) denotes the detuning from threshold, the
resonant peak grows only as \(\varepsilon^{-1/2}\), whereas the unsubtracted
threshold background grows as \(\varepsilon^{-3/2}\).  A symmetric
barrier--well--barrier cavity and a numerical asymmetric two-step attractive
well provide complementary controls.  The result is a new threshold-subtracted
stationary-clock candidate that separates universal threshold kinematics from
pole-sensitive resonant delay.
\end{abstract}

\maketitle
\newpage
\section{Introduction}

The question of how long a quantum particle spends in, or takes to cross, a
scattering region has no unique answer.  Different operational constructions
lead to different observables, including phase times, dwell times, traversal
times, and clock-based definitions
\cite{Hauge1989,ButtikerLandauer1982,LandauerMartin1994,
MugaLeavens2000,Winful2006,Field2022}.
Clock-based approaches include the Salecker--Wigner--Peres clock and related
Larmor-clock constructions
\cite{SaleckerWigner1958,Peres1980,Buttiker1983,Leavens1993,Davies2005,
Calcada2009,Lunardi2011}.
Stationary quantum-clock constructions provide a convenient way to probe
scattering-time observables without committing to a full time-dependent
wave-packet analysis. In one-dimensional scattering, however, threshold behavior can qualitatively distort the stationary clock reading, masking the resonant response that one would most like to isolate. The present paper is motivated by a simple question: can one separate the universal low-energy threshold contribution in the Salecker--Wigner--Peres (SWP) stationary clock time~\cite{SaleckerWigner1958,Peres1980,Davies2005,Sokolovski2017} from the genuinely resonant contribution associated with transmission resonances? Related threshold-singularity issues for dwell-time and phase-delay observables have been discussed in adjacent contexts~\cite{Kelkar2007}.

We answer this question in two stages. First, we prove a general low-energy theorem for the stationary Peres clock for real compactly supported one-dimensional potentials. In the generic sector, where there is no zero-energy resonance, the raw clock time contains a $1/\sqrt{E}$ threshold term whose coefficient is fixed by low-energy scattering data. In the exceptional zero-energy-resonant sector, the same scaling persists with a modified coefficient. This identifies the threshold obstruction at the level of general scattering theory rather than only in a special model.

Second, we work out the attractive square well in detail as an analytically
solvable laboratory for the subtraction strategy. In this model the threshold
coefficient, exact clock time, and threshold-subtracted observable can all be
written explicitly and compared directly with the dwell time and the
transmission Wigner phase delay. This lets us quantify the masking mechanism
and show how the resonant response is recovered once the threshold contribution
is removed. A symmetric barrier--well--barrier cavity and a numerical
asymmetric two-step attractive well then serve as complementary non-square
control examples.

Our aim is not to settle the broader tunneling-time debate or to identify a
unique preferred time observable~\cite{Hauge1989,LandauerMartin1994,
MugaLeavens2000,Winful2006,Field2022,Sokolovski2017}.
Rather, within the stationary Peres-clock phase-time convention, we identify a
universal low-energy threshold mask fixed by zero-energy scattering data. Once
this mask is subtracted, the resulting stationary-clock observable cleanly
exposes the resonant delay.

This paper is organized as follows. Section~\ref{sec:model} defines the square-well model and the stationary clock convention. Section~\ref{sec:general-threshold} states the general low-energy threshold theorem, whose transfer-matrix derivation is given in Appendix~\ref{app:low-energy}. Sections~\ref{sec:squarewell}--\ref{sec:resonances} work out the square well, compare with the dwell time and Wigner phase delay, and analyze the threshold masking of resonances. Sections~\ref{sec:bwb} and~\ref{sec:numerics} present complementary control examples: a barrier--well--barrier cavity and a numerical two-step well. We conclude in Sec.~\ref{sec:discussion}.

\section{Square-well setup and stationary clock definition}
\label{sec:model}

We consider the one-dimensional attractive square well
\begin{equation}
V(x)=
\begin{cases}
-\,V_0, & 0<x<a,\\
0, & x<0 \ \text{or}\ x>a,
\end{cases}
\qquad V_0>0,
\label{eq:V-squarewell}
\end{equation}
with stationary Schr\"odinger equation
\begin{equation}
-\psi''(x)+V(x)\psi(x)=E\psi(x),
\qquad E>0.
\label{eq:Sch-eq}
\end{equation}
We use
\begin{equation}
k=\sqrt{E},
\qquad
q=\sqrt{E+V_0},
\qquad
\kappa:=\sqrt{V_0}.
\label{eq:kq-def}
\end{equation}
For left incidence, the transmission amplitude is
\begin{equation}
t(E)=\frac{e^{-ika}}
{\cos(qa)-i\dfrac{k^2+q^2}{2kq}\sin(qa)}.
\label{eq:t-squarewell}
\end{equation}
We define the stationary clock phase by restoring the free phase across the interaction region,
\begin{equation}
\delta_P(E):=\arg\!\big[t(E)e^{ika}\big],
\qquad
\tau_P(E):=\dv{\delta_P}{E}.
\label{eq:Peres-def}
\end{equation}

\begin{remark}[Operational meaning]
The quantity \(\tau_P(E)\) should be understood as a stationary,
phase-sensitive scattering observable rather than as a single-event arrival
time. In principle, the phase \(\arg[t(E)e^{ikL}]\) can be extracted by
interferometric comparison of the transmitted wave with a reference wave, or
equivalently by reconstructing the complex transmission amplitude over a narrow
energy window. The clock time is then the energy derivative of this measured
phase. The threshold-subtracted clock \(\tau_{\rm sub}\) is operationally
defined by the same phase measurement together with the independently determined
low-energy threshold coefficient \(\ell_{\rm thr}\). Thus the subtraction removes a universal continuum-edge
contribution from a measurable stationary phase-time observable.
\end{remark}

In this paper, ``stationary SWP/Peres clock time'' refers to this stationary phase-time convention obtained from the phase of $t(E)e^{ika}$. We do not use alternative postselected regional-clock definitions in which the clock is coupled only to a specified subregion and the resulting quantity may coincide with a dwell-type time. The quantity $\tau_P(E)$ is the raw stationary SWP/Peres clock time whose threshold structure we seek to disentangle~\cite{Peres1980,Davies2005}.

\section{A general low-energy threshold theorem}
\label{sec:general-threshold}

Before turning back to the square well, it is useful to isolate the threshold structure of the stationary clock in a model-independent form. We therefore consider a general real compactly supported potential with
\begin{equation}
\operatorname{supp}V\subset[0,L].
\end{equation}
Writing $E=k^2$ with $k>0$, let $t(k)$ denote the transmission amplitude for left incidence. The stationary Peres phase and clock time are defined as in Sec.~\ref{sec:model}, with the free phase restored across the support length $L$:
\begin{equation}
\delta_P(E):=\arg\!\big[t(k)e^{ikL}\big],
\qquad
\tau_P(E):=\dv{\delta_P}{E}.
\label{eq:deltaP-body}
\end{equation}

To state the threshold law precisely, let
\begin{equation}
\mathsf T(k)=
\begin{pmatrix}
A(k) & B(k)\\
C(k) & D(k)
\end{pmatrix}
\end{equation}
denote the transfer matrix across the support, defined in Appendix~\ref{app:low-energy}. Since the Schr\"odinger equation depends on $k$ only through $k^2$, the entries are even functions of $k$ and admit expansions
\begin{align}
A(k)&=A_0+A_2k^2+A_4k^4+O(k^6), &
B(k)&=B_0+B_2k^2+B_4k^4+O(k^6), \nonumber\\
C(k)&=C_0+C_2k^2+C_4k^4+O(k^6), &
D(k)&=D_0+D_2k^2+D_4k^4+O(k^6),
\label{eq:ABCD-main-exp}
\end{align}
with real coefficients.

The relevant threshold distinction is between the generic sector $C_0\neq0$ and the exceptional sector $C_0=0$. As shown in Appendix~\ref{app:low-energy}, the condition $C_0=0$ is equivalent to the existence of a nontrivial bounded zero-energy solution on $\mathbb R$.

\begin{theorem}[Low-energy threshold law for the stationary Peres clock]
\label{thm:general-threshold-body}
Let $V$ be real and compactly supported in $[0,L]$.

\begin{enumerate}[label=(\roman*)]
\item \textbf{Generic sector.} If $C_0\neq0$, then
\begin{equation}
\delta_P(E)=\delta_* -2\ell_{\rm thr}\sqrt{E}+O(E^{3/2}),
\qquad
\tau_P(E)= -\frac{\ell_{\rm thr}}{\sqrt{E}}+O(\sqrt{E}),
\qquad E\to0^+,
\label{eq:tauP-generic-body}
\end{equation}
where
\begin{equation}
\ell_{\rm thr}:=-\frac{A_0+D_0}{2C_0},
\label{eq:ellthr-main}
\end{equation}
and $\delta_*$ is an $E$-independent phase determined by the sign of $C_0$ and the branch choice.

\item \textbf{Exceptional sector.} If $C_0=0$, then $A_0D_0=1$ and $A_0+D_0\neq0$. Define
\begin{equation}
t_0:=\frac{2}{A_0+D_0},
\qquad
\tilde\ell_{\rm thr}:=-\frac{B_0-C_2}{2(A_0+D_0)}.
\label{eq:elltildemain}
\end{equation}
Then
\begin{equation}
\delta_P(E)=\arg t_0 -2\tilde\ell_{\rm thr}\sqrt{E}+O(E^{3/2}),
\qquad
\tau_P(E)= -\frac{\tilde\ell_{\rm thr}}{\sqrt{E}}+O(\sqrt{E}),
\qquad E\to0^+.
\label{eq:tauP-exceptional-body}
\end{equation}
Here $t_0$ is real and nonzero. In special cases, such as the symmetric square-well half-bound tuning, one further has $t_0=\pm1$.
\end{enumerate}
\end{theorem}

\begin{corollary}[Generic threshold subtraction]
\label{cor:general-threshold-body}
In the generic sector, the threshold-subtracted clock
\begin{equation}
\tau_{\rm sub}(E):=\tau_P(E)+\frac{\ell_{\rm thr}}{\sqrt{E}}
\label{eq:tausub-body}
\end{equation}
is finite at threshold, in fact
\begin{equation}
\tau_{\rm sub}(E)=O(\sqrt{E}),
\qquad E\to0^+.
\label{eq:tausub-threshold-body}
\end{equation}
\end{corollary}

The proof follows from the low-energy transfer-matrix expansion given in Appendix~\ref{app:low-energy}.

\begin{remark}
Within the generic sector \(C_0\neq0\), the coefficient \(\ell_{\rm thr}\) may vanish on special parameter subsets. In that case the displayed \(1/\sqrt E\) contribution is absent and the next term in the low-energy expansion controls the threshold behavior. When referring to a generic threshold divergence, we assume \(\ell_{\rm thr}\neq0\).
\end{remark}

\begin{remark}[Barrier tunneling versus barrier-top crossing]
The threshold law is not restricted to attractive wells. It also applies to
positive barriers. For example, consider the square barrier
\[
V(x)=U,\qquad 0<x<a,\qquad U>0 .
\]
For fixed \(U>0\) and \(a>0\), the low-energy limit \(E\to0^+\) is an
under-barrier scattering problem: the asymptotic momentum \(k=\sqrt E\)
vanishes, while the interior wave number is imaginary. The zero-energy
Cauchy-data transfer matrix is
\[
T_0=
\begin{pmatrix}
\cosh(\mu a) & \mu^{-1}\sinh(\mu a)\\
\mu\sinh(\mu a) & \cosh(\mu a)
\end{pmatrix},
\qquad \mu=\sqrt U .
\]
Hence
\[
\ell_{\rm thr}
=
-\frac{A_0+D_0}{2C_0}
=
-\frac{\coth(\mu a)}{\mu},
\]
and the stationary Peres clock has the generic low-energy behavior
\[
\tau_P(E)
=
-\frac{\ell_{\rm thr}}{\sqrt E}
+O(\sqrt E)
=
\frac{\coth(\sqrt U\,a)}{\sqrt U\,\sqrt E}
+O(\sqrt E).
\]
This is a fixed-barrier low-energy asymptotic. It is not uniform in the
removed-barrier limits \(a\to0\) or \(U\to0\), where \(C_0\to0\) and the
problem leaves the generic \(C_0\neq0\) sector. The apparent divergence of the
coefficient in those limits therefore reflects noncommuting limits, not a
physical time for a nonexistent barrier.

This should be distinguished from the barrier-top limit \(E\to U\), where the
interior evanescent wave number vanishes but the finite-width transfer matrix
is regular. Thus the subtraction developed here removes the external
low-energy threshold contribution, not a singularity at the top of the barrier.
\end{remark}

\section{Exact square-well clock time and threshold subtraction}
\label{sec:squarewell}

We now specialize the general threshold theorem of Sec.~\ref{sec:general-threshold} to the attractive square well, where the transmission amplitude, stationary clock time, threshold coefficient, and subtracted clock can all be written explicitly. The square well therefore provides a useful analytic laboratory for exhibiting the general threshold obstruction and its removal.

\begin{proposition}[Exact stationary clock time]
For the attractive square well, the stationary Peres clock time is
\begin{equation}
\tau_P(E)=
\frac{
 a k^2 q\,(k^2+q^2)-\dfrac12 (q^2-k^2)^2\sin(2qa)
}{
 kq\Big(4k^2q^2+(q^2-k^2)^2\sin^2(qa)\Big)
}.
\label{eq:tauP-exact}
\end{equation}
\end{proposition}

\begin{proof}
Write
\begin{equation}
t(E)e^{ika}=\frac{1}{A-iB},
\qquad
A=\cos(qa),
\qquad
B=\frac{k^2+q^2}{2kq}\sin(qa).
\end{equation}
Then
\begin{equation}
\tau_P(E)=\frac{A\,B'-B\,A'}{A^2+B^2},
\end{equation}
where primes denote derivatives with respect to $E$. Straightforward algebra yields Eq.~\eqref{eq:tauP-exact}.
\end{proof}

In the language of Theorem~\ref{thm:general-threshold-body}, the generic square-well sector is characterized by $\sin(\kappa a)\neq 0$, while the exceptional zero-energy-resonant / half-bound sector is characterized by $\sin(\kappa a)=0$.

The exact expression immediately reveals two distinct low-energy regimes: the generic case $\sin(\kappa a)\neq 0$ and the threshold-resonant / half-bound tuning $\sin(\kappa a)=0$.

\begin{proposition}[Generic threshold law]
Assume $\sin(\kappa a)\neq 0$. Then as $E\to0^+$,
\begin{equation}
\tau_P(E)=
-\frac{\cot(\kappa a)}{\kappa\sqrt{E}}
+O(\sqrt{E}).
\label{eq:tauP-threshold}
\end{equation}
Equivalently,
\begin{equation}
\tau_P(E)\sim \frac{C_{\rm thr}}{\sqrt{E}},
\qquad
C_{\rm thr}=-\frac{\cot(\kappa a)}{\kappa}.
\label{eq:Cthr-def}
\end{equation}
\end{proposition}

\begin{proof}
Expanding Eq.~\eqref{eq:tauP-exact} at small $k=\sqrt{E}$ using
\begin{equation}
q=\sqrt{\kappa^2+k^2}=\kappa+\frac{k^2}{2\kappa}+O(k^4)
\end{equation}
gives Eq.~\eqref{eq:tauP-threshold}.
\end{proof}

The leading singular term in Eq.~\eqref{eq:tauP-threshold} is the threshold obstruction in the raw stationary clock. This motivates the threshold-subtracted clock time
\begin{equation}
\tau_{\rm sub}(E):=
\tau_P(E)+\frac{\cot(\kappa a)}{\kappa\sqrt{E}},
\qquad \sin(\kappa a)\neq 0.
\label{eq:tausub-def}
\end{equation}

\begin{corollary}[Finite threshold-subtracted clock time]
For $\sin(\kappa a)\neq 0$,
\begin{equation}
\tau_{\rm sub}(E)=O(\sqrt{E}),
\qquad E\to0^+,
\end{equation}
and in particular
\begin{equation}
\tau_{\rm sub}(E)\to 0
\qquad (E\to0^+).
\label{eq:tausub-vanish}
\end{equation}
\end{corollary}

The exceptional threshold-resonant / half-bound tuning is handled separately.

\begin{proposition}[Threshold-resonant case]
If $\sin(\kappa a)=0$, then the generic subtraction \eqref{eq:tausub-def} is not the correct one. Instead,
\begin{equation}
\tau_P(E)=\frac{a}{8\sqrt{E}}+O(\sqrt{E})
\qquad (E\to 0^+).
\label{eq:tauP-halfbound}
\end{equation}
\end{proposition}

\begin{proof}
At the threshold-resonant tuning $\sin(\kappa a)=0$, one has $\kappa a=n\pi$ for some integer $n$. Write $k=\sqrt{E}$ and
\[
q=\sqrt{\kappa^2+k^2}
=\kappa+\frac{k^2}{2\kappa}+O(k^4).
\]
Then
\[
qa=n\pi+\frac{a k^2}{2\kappa}+O(k^4),
\]
so
\[
\sin(qa)=(-1)^n\frac{a k^2}{2\kappa}+O(k^4),
\qquad
\sin(2qa)=\frac{a k^2}{\kappa}+O(k^4).
\]
Also,
\[
q^2-k^2=\kappa^2
\]
exactly, and
\[
q=\kappa+O(k^2).
\]
Substituting into Eq.~\eqref{eq:tauP-exact}, the numerator becomes
\[
a k^2 q(k^2+q^2)-\frac12(q^2-k^2)^2\sin(2qa)
=
\frac12 a\kappa^3 k^2+O(k^4),
\]
while the denominator is
\[
kq\Bigl(4k^2q^2+(q^2-k^2)^2\sin^2(qa)\Bigr)
=
4\kappa^3 k^3+O(k^5).
\]
Therefore
\[
\tau_P(E)=\frac{a}{8k}+O(k)
=\frac{a}{8\sqrt{E}}+O(\sqrt{E}),
\]
which proves Eq.~\eqref{eq:tauP-halfbound}.
\end{proof}

A useful amplitude-level perspective is obtained by expanding the transmission factor itself. For $\sin(\kappa a)\neq 0$,
\begin{equation}
t(E)e^{ika}
=
\frac{2ik}{\kappa\sin(\kappa a)}
\left[
1-2i\frac{\cot(\kappa a)}{\kappa}k+O(k^2)
\right],
\label{eq:amp-threshold-expansion}
\end{equation}
so the divergent phase derivative is directly tied to the low-energy threshold factor.

\section{Comparison with dwell time and Wigner phase delay}
\label{sec:comparison}

The interior wavefunction can be written as
\begin{equation}
\psi_{II}(x)=A\,e^{iqx}+B\,e^{-iqx},
\label{eq:psiII-def}
\end{equation}
and the dwell time is defined by~\cite{Smith1960}
\begin{equation}
\tau_D(E):=\frac{1}{j_{\rm in}}\int_0^a \dd x\,|\psi_{II}(x)|^2,
\qquad j_{\rm in}=2k.
\label{eq:tauD-def}
\end{equation}
In the units of Eq.~\eqref{eq:Sch-eq}, the probability current is
$j=2\,\mathrm{Im}(\psi^*\psi')$, so a unit-amplitude incoming plane wave has
$j_{\rm in}=2k$. A direct calculation gives
\begin{equation}
\tau_D(E)=
\frac{
 k\left[2a(k^2+q^2)+\dfrac{q^2-k^2}{q}\sin(2qa)\right]
}{
 2\Big(4k^2q^2+(q^2-k^2)^2\sin^2(qa)\Big)
}.
\label{eq:tauD-exact}
\end{equation}

\begin{proposition}[Threshold behavior of the dwell time]
For $\sin(\kappa a)\neq 0$,
\begin{equation}
\tau_D(E)=
\sqrt{E}\,
\frac{2a\kappa+\sin(2\kappa a)}
{2\kappa^3\sin^2(\kappa a)}
+O(E^{3/2}),
\qquad E\to0^+.
\label{eq:tauD-threshold}
\end{equation}
Hence the dwell time is threshold-benign: it vanishes as $\sqrt{E}$. At the threshold-resonant tuning $\sin(\kappa a)=0$,
\begin{equation}
\tau_D(E)=\frac{a}{4\sqrt{E}}+O(\sqrt{E}).
\label{eq:tauD-halfbound}
\end{equation}
\end{proposition}

\begin{proof}
Using Eq.~\eqref{eq:tauD-exact} with $k=\sqrt{E}$ and $q=\sqrt{\kappa^2+k^2}$, first consider the generic case $\sin(\kappa a)\neq0$. As $k\to0$,
\[
q=\kappa+\frac{k^2}{2\kappa}+O(k^4),
\qquad
q^2-k^2=\kappa^2,
\qquad
q^2+k^2=\kappa^2+2k^2,
\]
and
\[
\sin(2qa)=\sin(2\kappa a)+O(k^2),
\qquad
\sin^2(qa)=\sin^2(\kappa a)+O(k^2).
\]
Therefore the numerator becomes
\[
k\left[2a(k^2+q^2)+\frac{q^2-k^2}{q}\sin(2qa)\right]
=
k\left[2a\kappa^2+\kappa\sin(2\kappa a)\right]
+O(k^3),
\]
while the denominator is
\[
2\Big(4k^2q^2+(q^2-k^2)^2\sin^2(qa)\Big)
=
2\kappa^4\sin^2(\kappa a)+O(k^2).
\]
Hence
\[
\tau_D(E)=
k\,\frac{2a\kappa^2+\kappa\sin(2\kappa a)}
{2\kappa^4\sin^2(\kappa a)}
+O(k^3)
=
\sqrt{E}\,
\frac{2a\kappa+\sin(2\kappa a)}
{2\kappa^3\sin^2(\kappa a)}
+O(E^{3/2}),
\]
which is Eq.~\eqref{eq:tauD-threshold}.

Now consider the threshold-resonant / half-bound case $\sin(\kappa a)=0$, so $\kappa a=n\pi$ for some integer $n$. Then
\[
q=\kappa+\frac{k^2}{2\kappa}+O(k^4),
\qquad
qa=n\pi+\frac{a k^2}{2\kappa}+O(k^4),
\]
and therefore
\[
\sin(2qa)=\frac{a k^2}{\kappa}+O(k^4),
\qquad
\sin^2(qa)=\frac{a^2 k^4}{4\kappa^2}+O(k^6).
\]
Substituting into Eq.~\eqref{eq:tauD-exact}, the numerator becomes
\[
k\left[2a(k^2+q^2)+\frac{q^2-k^2}{q}\sin(2qa)\right]
=
2a\kappa^2 k+O(k^3),
\]
while the denominator is
\[
2\Big(4k^2q^2+(q^2-k^2)^2\sin^2(qa)\Big)
=
8\kappa^2 k^2+O(k^4).
\]
Hence
\[
\tau_D(E)=\frac{a}{4k}+O(k)
=
\frac{a}{4\sqrt{E}}+O(\sqrt{E}),
\]
which is Eq.~\eqref{eq:tauD-halfbound}.
\end{proof}

Thus the dwell time and the raw stationary clock behave very differently at threshold: the former generically vanishes as $\sqrt{E}$, while the latter diverges as $1/\sqrt{E}$.

For comparison, define the transmission Wigner phase delay by~\cite{Wigner1955,Smith1960}
\begin{equation}
\tau_W(E):=\dv{}{E}\arg t(E).
\label{eq:tauW-def}
\end{equation}
Since
\begin{equation}
\arg\!\big[t(E)e^{ika}\big]=\arg t(E)+ka,
\end{equation}
we have
\begin{equation}
\tau_W(E)=\tau_P(E)-\frac{a}{2\sqrt{E}}.
\label{eq:tauW-tauP}
\end{equation}

Thus the transmission Wigner phase delay also inherits a threshold divergence, reinforcing that the low-energy singularity is not unique to the raw stationary clock.

\section{Transmission resonances and threshold masking}
\label{sec:resonances}

The transmission probability is
\begin{equation}
T(E)=|t(E)|^2=
\frac{1}{1+\dfrac{V_0^2}{4E(E+V_0)}\sin^2(qa)}.
\label{eq:T-exact}
\end{equation}
Its maxima occur at
\begin{equation}
qa=n\pi,
\qquad
E_n=\frac{n^2\pi^2}{a^2}-V_0,
\qquad E_n>0.
\label{eq:En-def}
\end{equation}
We refer to $E_n$ as the transmission-resonance energies of the square well.

\subsection{Local resonant form of the subtracted clock}

Write
\begin{equation}
E=E_n+\Delta E,
\qquad
q_n:=\sqrt{E_n+V_0}=\frac{n\pi}{a},
\qquad
k_n:=\sqrt{E_n}.
\end{equation}
Near $E_n$,
\begin{equation}
qa=n\pi+\frac{a}{2q_n}\Delta E+O(\Delta E^2),
\end{equation}
so
\begin{equation}
\sin(qa)\approx (-1)^n\frac{a}{2q_n}\Delta E,
\qquad
\cos(qa)\approx (-1)^n.
\end{equation}
Expanding the exact amplitude near $E_n$ gives the local arctangent phase profile and hence the following result.

\begin{proposition}[Local resonant form of the subtracted clock time]
Assume \(\sin(\kappa a)\neq0\), so that the generic subtraction
\eqref{eq:tausub-def} is defined. Let \(E_n>0\) satisfy
Eq.~\eqref{eq:En-def}. Then near \(E=E_n\),
\begin{equation}
\tau_{\rm sub}(E)\approx
\frac{H_n}{1+H_n^2(E-E_n)^2}
+\text{smooth background},
\label{eq:tausub-Lorentz}
\end{equation}
where
\begin{equation}
H_n=\frac{a(2E_n+V_0)}{4\sqrt{E_n}(E_n+V_0)}.
\label{eq:Hn-def}
\end{equation}
Equivalently,
\begin{equation}
\tau_{\rm sub}(E)\approx
\frac{\Gamma_n/2}{(E-E_n)^2+(\Gamma_n/2)^2}
+\text{smooth background},
\qquad
\Gamma_n=\frac{8\sqrt{E_n}(E_n+V_0)}{a(2E_n+V_0)}.
\label{eq:Gamma-def}
\end{equation}
\end{proposition}

\begin{proof}
Assume \(\sin(\kappa a)\neq 0\), so that \(\tau_{\rm sub}\) is defined by
Eq.~\eqref{eq:tausub-def}. Let \(E_n>0\) satisfy \(q_n a=n\pi\), with
\[
q_n=\sqrt{E_n+V_0}=\frac{n\pi}{a},
\qquad
k_n=\sqrt{E_n}.
\]
Write
\[
E=E_n+\Delta E.
\]
Since \(q(E)=\sqrt{E+V_0}\), one has
\[
q=q_n+\frac{\Delta E}{2q_n}+O(\Delta E^2),
\qquad
qa=n\pi+\frac{a}{2q_n}\Delta E+O(\Delta E^2).
\]
Therefore
\[
\sin(qa)=(-1)^n\frac{a}{2q_n}\Delta E+O(\Delta E^2),
\qquad
\cos(qa)=(-1)^n+O(\Delta E^2).
\]

Using
\[
t(E)e^{ika}=\frac{1}{A-iB},
\qquad
A=\cos(qa),
\qquad
B=\frac{k^2+q^2}{2kq}\sin(qa),
\]
and noting that the prefactor in \(B\) is smooth near \(E_n\), we may evaluate it to first order at \(E_n\):
\[
\frac{k^2+q^2}{2kq}
=
\frac{2E_n+V_0}{2k_n q_n}+O(\Delta E).
\]
Hence
\[
B=
(-1)^n\frac{2E_n+V_0}{2k_n q_n}\frac{a}{2q_n}\Delta E
+O(\Delta E^2)
=
(-1)^n H_n\,\Delta E+O(\Delta E^2),
\]
where
\[
H_n=\frac{a(2E_n+V_0)}{4\sqrt{E_n}(E_n+V_0)}.
\]
Since also
\[
A=(-1)^n+O(\Delta E^2),
\]
it follows that
\[
A-iB
=
(-1)^n\Bigl(1-iH_n\Delta E+O(\Delta E^2)\Bigr).
\]
Therefore, on a fixed smooth phase branch,
\[
\delta_P(E)
=
\arg\!\bigl[t(E)e^{ika}\bigr]
=
\delta_n+\arctan\!\bigl(H_n\Delta E\bigr)+O(\Delta E^2),
\]
for some constant phase \(\delta_n\). Differentiating with respect to \(E\)
gives
\[
\tau_P(E)
=
\frac{H_n}{1+H_n^2(E-E_n)^2}
+\text{smooth background}.
\]

Since the subtraction term
\[
\frac{\cot(\kappa a)}{\kappa\sqrt{E}}
\]
is smooth near every fixed resonance energy \(E_n>0\), passing from
\(\tau_P\) to \(\tau_{\rm sub}\) only changes the smooth background term.
Thus
\[
\tau_{\rm sub}(E)\approx
\frac{H_n}{1+H_n^2(E-E_n)^2}
+\text{smooth background}.
\]

Finally, defining
\[
\Gamma_n:=\frac{2}{H_n}
=
\frac{8\sqrt{E_n}(E_n+V_0)}{a(2E_n+V_0)},
\]
one rewrites the Lorentzian as
\[
\frac{H_n}{1+H_n^2(E-E_n)^2}
=
\frac{\Gamma_n/2}{(E-E_n)^2+(\Gamma_n/2)^2},
\]
which proves Eqs.~\eqref{eq:tausub-Lorentz} and \eqref{eq:Gamma-def}.
\end{proof}

\subsection{Dwell-time resonance comparison}

Using Eq.~\eqref{eq:tauD-exact} and freezing the slowly varying numerator at $E_n$ while keeping the quadratic resonant dip in the denominator gives
\begin{equation}
\tau_D(E)\approx
\frac{\tau_{D,n}^{\max}}
{1+\left(\dfrac{E-E_n}{\Delta_n/2}\right)^2}
+\text{smooth background},
\label{eq:tauD-Lorentz}
\end{equation}
with peak height
\begin{equation}
\tau_{D,n}^{\max}=
\frac{a(2E_n+V_0)}{4\sqrt{E_n}(E_n+V_0)},
\label{eq:tauD-peak}
\end{equation}
and width scale
\begin{equation}
\Delta_n=
\frac{8\sqrt{E_n}(E_n+V_0)}{aV_0}.
\label{eq:Delta-def}
\end{equation}
In particular,
\begin{equation}
\tau_{D,n}^{\max}=\tau_{{\rm sub},n}^{\max},
\qquad
\frac{\Delta_n}{\Gamma_n}=\frac{2E_n+V_0}{V_0}.
\label{eq:peak-relations}
\end{equation}
For a resonance approaching threshold, $E_n\ll V_0$, the widths coincide asymptotically and both peak heights scale as $E_n^{-1/2}$.

This comparison clarifies the role of the subtraction: the threshold-subtracted stationary clock does not coincide globally with the dwell time, but near an isolated square-well resonance it exposes the same leading resonant peak height and the same near-threshold scaling, up to smooth background terms.

\subsection{Quantitative threshold masking}

To exhibit the masking mechanism explicitly, write
\begin{equation}
\kappa a=n\pi-\varepsilon,
\qquad
0<\varepsilon\ll1.
\label{eq:epsilon-def}
\end{equation}
Then the corresponding resonance lies at
\begin{equation}
E_n\approx \frac{2\kappa\varepsilon}{a}.
\label{eq:En-nearthresh}
\end{equation}
At the resonance energy,
\begin{equation}
\tau_{{\rm sub},n}^{\max}\sim \tau_{D,n}^{\max}\sim \varepsilon^{-1/2},
\label{eq:peak-scaling}
\end{equation}
whereas the raw threshold background scales as
\begin{equation}
|\tau_{\rm thr}(E_n)|\sim \varepsilon^{-3/2}.
\label{eq:thr-scaling}
\end{equation}
Hence
\begin{equation}
\frac{\tau_{{\rm sub},n}^{\max}}{|\tau_{\rm thr}(E_n)|}
\sim
\frac{\tau_{D,n}^{\max}}{|\tau_{\rm thr}(E_n)|}
\sim \varepsilon,
\qquad
\varepsilon\to0^+.
\label{eq:masking}
\end{equation}
Equation~\eqref{eq:masking} is the quantitative statement that the unsubtracted stationary clock becomes progressively worse at resolving a resonance as that resonance approaches threshold, even though the resonant peak itself is growing.

\section{Control example: local resonant structure in a symmetric barrier--well--barrier cavity}
\label{sec:bwb}

As a complementary control example, we consider a symmetric
barrier--well--barrier cavity~\cite{Hartman1962,Calcada2009,Lunardi2011}.
This geometry is not needed for the threshold theorem, but it connects naturally
to the tunneling-time and resonant-tunneling setting and makes the local
Breit--Wigner structure especially transparent once the smooth threshold
background is separated off.

Consider
\begin{equation}
V(x)=
\begin{cases}
U, & 0<x<b,\\
-\,V_0, & b<x<b+a,\\
U, & b+a<x<2b+a,\\
0, & x<0 \ \text{or}\ x>2b+a,
\end{cases}
\qquad
U>0,
\quad
V_0>0,
\label{eq:bwb-potential}
\end{equation}
and restrict for definiteness to
\begin{equation}
0<E<U,
\qquad
k=\sqrt{E},
\qquad
q=\sqrt{E+V_0},
\qquad
\kappa_B=\sqrt{U-E}.
\label{eq:bwb-scales}
\end{equation}

\subsection{Exact cavity composition from the transfer matrix}

To avoid any ambiguity associated with mismatched asymptotic wavenumbers, we formulate the cavity exactly in terms of the scattering amplitudes of a \emph{single} barrier with distinct left- and right-incidence amplitudes. Let the left barrier occupy $0<x<b$, separating the exterior medium of wavenumber $k$ from the cavity medium of wavenumber $q$. Denote the corresponding amplitudes by
\begin{equation}
r_L(E),\qquad t_L(E),\qquad r_R(E),\qquad t_R(E),
\label{eq:singlebarrier-amps}
\end{equation}
where $r_L$ and $t_L$ refer to incidence from the exterior side, while $r_R$ and $t_R$ refer to incidence from the cavity side. In terms of incoming/outgoing amplitudes,
\begin{equation}
b_L=r_L a_L+t_R c_-,
\qquad
c_+=t_L a_L+r_R c_-.
\label{eq:leftbarrier-scatt}
\end{equation}
Equivalently, the exact transfer matrix of the left barrier is
\begin{equation}
\binom{c_+}{c_-}
=
M_L(E)
\binom{a_L}{b_L},
\qquad
M_L(E)=\frac{1}{t_R(E)}
\begin{pmatrix}
t_L t_R-r_L r_R & r_R \\
-r_L & 1
\end{pmatrix}.
\label{eq:ML-def}
\end{equation}

The right barrier is the mirror image of the left one. Writing $d_+$ and $d_-$ for the right- and left-moving amplitudes at the left edge of the right barrier and $(b_R,a_R)$ for the outgoing/incoming amplitudes in the exterior region on the far right, its exact transfer matrix is
\begin{equation}
\binom{b_R}{a_R}
=
M_R(E)
\binom{d_+}{d_-},
\qquad
M_R(E)=\frac{1}{t_L(E)}
\begin{pmatrix}
t_L t_R-r_L r_R & r_L \\
-r_R & 1
\end{pmatrix}.
\label{eq:MR-def}
\end{equation}
Propagation across the central well of width $a$ is described by
\begin{equation}
\binom{d_+}{d_-}
=
P_q(a)
\binom{c_+}{c_-},
\qquad
P_q(a)=
\begin{pmatrix}
e^{iqa} & 0 \\
0 & e^{-iqa}
\end{pmatrix}.
\label{eq:Pq-def}
\end{equation}
Thus the full cavity transfer matrix is exactly
\begin{equation}
\binom{b_R}{a_R}=M_R(E)P_q(a)M_L(E)\binom{a_L}{b_L}.
\label{eq:full-bwb-transfer}
\end{equation}

For left incidence on the full cavity, one sets
\begin{equation}
a_L=1,
\qquad
a_R=0,
\qquad
b_R=t_{\rm bwb}(E),
\label{eq:left-inc-cavity}
\end{equation}
with $b_L$ the total reflection amplitude. Eliminating the internal amplitudes gives the exact composition law
\begin{equation}
t_{\rm bwb}(E)=
\frac{t_L(E)t_R(E)e^{iqa}}
{1-r_R(E)^2 e^{2iqa}}.
\label{eq:t-bwb-exact}
\end{equation}
This is the correct cavity formula for the symmetric barrier--well--barrier structure. The crucial point is that the numerator involves the distinct left- and right-incident single-barrier transmission amplitudes, while the denominator involves the cavity-side reflection amplitude $r_R$.

It is useful to parameterize
\begin{equation}
r_R(E)=\sqrt{\rho(E)}\,e^{i\phi_R(E)},
\qquad
0\le \rho(E)<1,
\label{eq:rR-polar}
\end{equation}
and
\begin{equation}
t_L(E)t_R(E)=\mathcal T(E)e^{i\chi(E)},
\label{eq:tLtR-polar}
\end{equation}
so that
\begin{equation}
t_{\rm bwb}(E)=
\frac{\mathcal T(E)e^{i(\chi+qa)}}{1-\rho(E)e^{i\Theta(E)}},
\qquad
\Theta(E):=2qa+2\phi_R(E).
\label{eq:t-bwb-Theta-correct}
\end{equation}

Define the stationary clock phase by restoring the free phase across the support $L:=2b+a$,
\begin{equation}
\delta_P^{\rm bwb}(E):=\arg\!\big[t_{\rm bwb}(E)e^{ikL}\big].
\label{eq:deltaP-bwb}
\end{equation}
Then
\begin{equation}
\delta_P^{\rm bwb}(E)=
\delta_{\rm sm}(E)
-\arg\!\big[1-\rho(E)e^{i\Theta(E)}\big],
\label{eq:delta-bwb-decomp-correct}
\end{equation}
where
\begin{equation}
\delta_{\rm sm}(E):=\chi(E)+q(E)a+k(E)L
\label{eq:delta-sm-bwb}
\end{equation}
is smooth on the scale of an isolated resonance.

\subsection{Local resonant form}

Transmission resonances occur at the cavity quantization condition
\begin{equation}
\Theta(E_n)=2\pi n,
\qquad n\in\mathbb Z.
\label{eq:bwb-res-cond-correct}
\end{equation}
Let
\begin{equation}
\rho_n:=\rho(E_n),
\qquad
\Theta_n':=\eval{\dv{\Theta}{E}}_{E_n}.
\label{eq:rho-theta-prime}
\end{equation}

Near \(E=E_n\), one has
\begin{equation}
1-\rho(E)e^{i\Theta(E)}
=
(1-\rho_n)-\rho'_n(E-E_n)
-i\rho_n\Theta'_n(E-E_n)
+O\!\left((E-E_n)^2\right).
\label{eq:bwb-den-exp-correct}
\end{equation}
Under the slow-variation assumption, the real linear term gives only a
subleading asymmetric correction to the local line shape and is neglected at
leading Breit--Wigner order.

\begin{proposition}[Local resonant form in the cavity geometry]
Let $E_n$ satisfy Eq.~\eqref{eq:bwb-res-cond-correct}, and assume the resonance is isolated and the barrier amplitudes vary slowly across the resonant window. Then
\begin{equation}
\tau_P^{\rm bwb}(E)
=
\tau_{\rm sm}^{\rm bwb}(E)
+
\operatorname{sgn}(\Theta_n')\,
\frac{\Gamma_n/2}{(E-E_n)^2+(\Gamma_n/2)^2}
+O(E-E_n),
\label{eq:tauP-bwb-BW-safe}
\end{equation}
where
\begin{equation}
\tau_{\rm sm}^{\rm bwb}(E):=\dv{\delta_{\rm sm}}{E}
\end{equation}
is a smooth background term and
\begin{equation}
\Gamma_n:=
\frac{2(1-\rho_n)}{\rho_n\,|\Theta_n'|}>0
\label{eq:Gamma-bwb-safe}
\end{equation}
is the local Breit--Wigner width parameter.
\end{proposition}

\begin{proof}
Equation~\eqref{eq:delta-bwb-decomp-correct} shows that the singular resonant behavior comes entirely from the cavity denominator. Using Eq.~\eqref{eq:bwb-den-exp-correct},
\begin{equation}
-\arg\!\big[1-\rho(E)e^{i\Theta(E)}\big]
\approx
\arctan\!\left(
\frac{\rho_n\Theta_n'}{1-\rho_n}(E-E_n)
\right),
\end{equation}
and differentiation gives a Lorentzian contribution with overall sign $\operatorname{sgn}(\Theta_n')$. Writing the result in terms of the manifestly positive width \eqref{eq:Gamma-bwb-safe} yields Eq.~\eqref{eq:tauP-bwb-BW-safe}.
\end{proof}

In the opaque-barrier regime $1-\rho_n\ll1$, one has $\Gamma_n\ll1$ and the resonance becomes parametrically narrow. Thus the barrier--well--barrier geometry provides a clean control example in which the stationary clock time exhibits an essentially textbook local Breit--Wigner structure.

\subsection{Generic threshold subtraction for the cavity}

The same compact-support theorem applies to the full barrier--well--barrier
potential. Thus, in the generic sector, the leading threshold contribution is
not heuristic: it is fixed by the zero-energy Cauchy-data transfer matrix of the
full cavity.

Let
\[
\mu:=\sqrt U,\qquad
\kappa:=\sqrt{V_0},\qquad
\alpha:=\mu b,\qquad
\theta:=\kappa a.
\]
At \(E=0\), the Cauchy-data transfer matrices for one barrier and for the
central well are
\[
M_B=
\begin{pmatrix}
\cosh\alpha & \mu^{-1}\sinh\alpha\\
\mu\sinh\alpha & \cosh\alpha
\end{pmatrix},
\qquad
M_W=
\begin{pmatrix}
\cos\theta & \kappa^{-1}\sin\theta\\
-\kappa\sin\theta & \cos\theta
\end{pmatrix}.
\]
The full zero-energy transfer matrix is
\[
\mathsf T_0=M_BM_WM_B
=
\begin{pmatrix}
A_0&B_0\\
C_0&D_0
\end{pmatrix}.
\]
Here \(\mathsf T_0\) is the Cauchy-data transfer matrix appearing in
Theorem~\ref{thm:general-threshold-body}, not the incoming/outgoing
amplitude transfer matrix used above to compose the finite-energy cavity.
For the symmetric cavity, \(A_0=D_0\), with
\[
A_0=D_0
=
\cos\theta\,\cosh(2\alpha)
+
\frac{\mu^2-\kappa^2}{2\mu\kappa}
\sin\theta\,\sinh(2\alpha),
\]
and
\[
C_0
=
\mu\cos\theta\,\sinh(2\alpha)
+
\sin\theta
\left(
\frac{\mu^2}{\kappa}\sinh^2\alpha
-\kappa\cosh^2\alpha
\right).
\]
Therefore, away from a full-cavity zero-energy resonance, \(C_0\neq0\), the
generic threshold coefficient is
\[
\ell_{\rm thr}^{\rm bwb}
=
-\frac{A_0+D_0}{2C_0}
=
-\frac{A_0}{C_0}.
\]
The corresponding threshold-subtracted cavity clock in the generic sector is
\[
\tau_{\rm sub}^{\rm bwb}(E)
:=
\tau_P^{\rm bwb}(E)
+
\frac{\ell_{\rm thr}^{\rm bwb}}{\sqrt E}.
\]
Near any fixed isolated resonance \(E_n>0\), this subtraction changes only the
smooth background. The local Breit--Wigner contribution in
Eq.~\eqref{eq:tauP-bwb-BW-safe} is therefore unchanged, while the leading
\(1/\sqrt{E}\) continuum-edge threshold term is removed.

At an exceptional full-cavity zero-energy resonance, \(C_0=0\), this generic
subtraction must be replaced by the exceptional coefficient from
Theorem~\ref{thm:general-threshold-body}, which depends on the next
low-energy transfer-matrix data.

\section{Numerical control beyond the square well}
\label{sec:numerics}

To test whether the threshold-subtraction mechanism persists beyond the exactly solvable square well, we considered an asymmetric two-step attractive well,
\begin{equation}
V(x)=
\begin{cases}
-\,V_1, & 0<x<a_1,\\
-\,V_2, & a_1<x<a_1+a_2,\\
0, & x<0 \ \text{or}\ x>a_1+a_2,
\end{cases}
\qquad
V_1>0,\quad V_2>0,
\label{eq:V-twostep}
\end{equation}
with $V_1\neq V_2$ and $a_1\neq a_2$ in general. The scattering amplitudes and stationary observables were computed numerically with a transfer-matrix method for piecewise-constant potentials, and the implementation was first calibrated against the exact square-well formulas.

For the square-well calibration case, the numerics reproduce the analytic transmission amplitude together with the stationary observables $\tau_P$, $\tau_W$, and $\tau_D$ to high accuracy, and recover both the generic and exceptionally tuned threshold coefficients. In particular, the maximum flux-conservation error is at the level of machine precision, the transmission amplitude is reproduced to $\lesssim 10^{-14}$, the dwell time to machine precision, and the phase-derived quantities $\tau_P$ and $\tau_W$ to the expected derivative-level accuracy. The fitted generic and exceptional threshold coefficients agree with the analytic values to better than $10^{-4}$ relative accuracy. Figure~\ref{fig:square-validation} shows the square-well calibration. At plotting resolution, the numerical and analytic curves are visually indistinguishable; the agreement is quantified in the error panel.

\begin{figure}[t]
    \centering
    \includegraphics[width=0.98\linewidth]{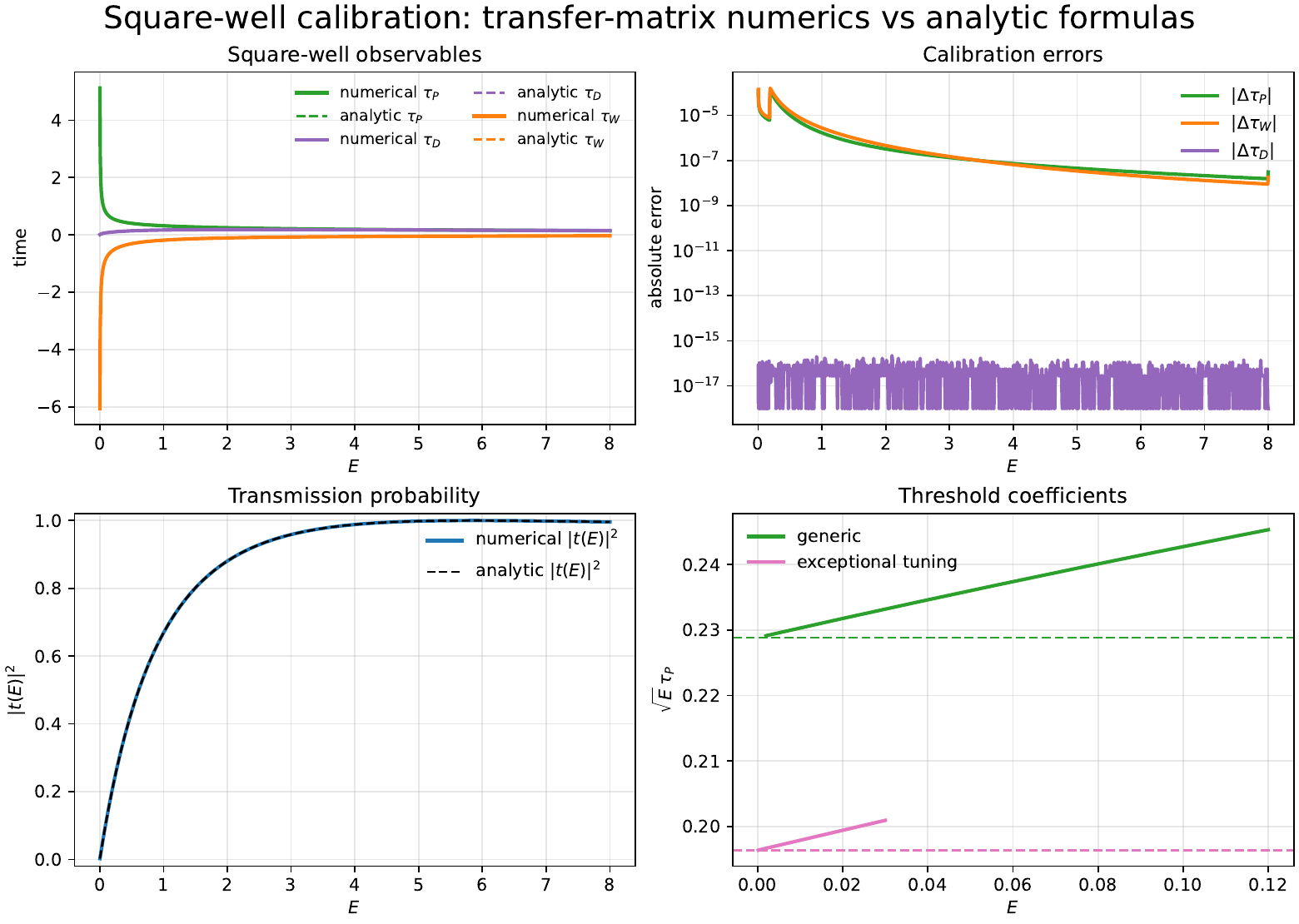}
    \caption{Square-well calibration of the transfer-matrix numerics against the exact analytic formulas. The upper-left panel compares the stationary observables $\tau_P$, $\tau_W$, and $\tau_D$, the upper-right panel shows the absolute numerical--analytic errors, the lower-left panel shows the transmission probability, and the lower-right panel shows the extraction of the generic and exceptionally tuned threshold coefficients.}
    \label{fig:square-validation}
\end{figure}

For the asymmetric two-step well, we use the default control parameters
$V_1=5.5$, $V_2=2.3$, $a_1=0.75$, and $a_2=1.30$.
No closed-form threshold coefficient is assumed. Instead, we define a numerical threshold-subtracted quantity by fitting the low-energy behavior of the raw stationary clock to
\begin{equation}
\tau_P(E)\approx \frac{C_{\rm thr}^{\rm fit}}{\sqrt{E}}+c_1\sqrt{E}
\label{eq:tauP-fit}
\end{equation}
over a low-energy fitting window and then setting
\begin{equation}
\tau_{\rm sub}^{\rm num}(E):=
\tau_P(E)-\frac{C_{\rm thr}^{\rm fit}}{\sqrt{E}}.
\label{eq:tausub-num}
\end{equation}
For this parameter choice, the fitted coefficient is $C_{\rm thr}^{\rm fit}=-0.783184$, with fit quality $R^2=0.999551$. The fitted coefficient changes by only $5.863\times10^{-6}$ under refinement from the default to the refined energy grid and by only $1.983\times10^{-3}$ under a modest variation of the fitting window. Over the near-threshold zoom window, the resulting $\tau_{\rm sub}^{\rm num}(E)$ changes by only about $2.7\%$ under that fit-window variation. The maximum flux-conservation error remains at the level of machine precision.

The numerical studies are not used as a proof beyond the compact-support theorem established above. Their role is more modest and more direct: they show that, for a genuinely non-square attractive structure, the raw stationary clock is again dominated by a strong threshold background, while subtracting the fitted leading term leaves a much milder quantity whose near-threshold structure is more readily analyzable. In this sense, the two-step well serves as a numerical control example corroborating that the subtraction mechanism is not merely a square-well artifact. Figure~\ref{fig:two-step-control} displays the corresponding non-square control example.

\begin{figure}[t]
    \centering
    \includegraphics[width=0.98\linewidth]{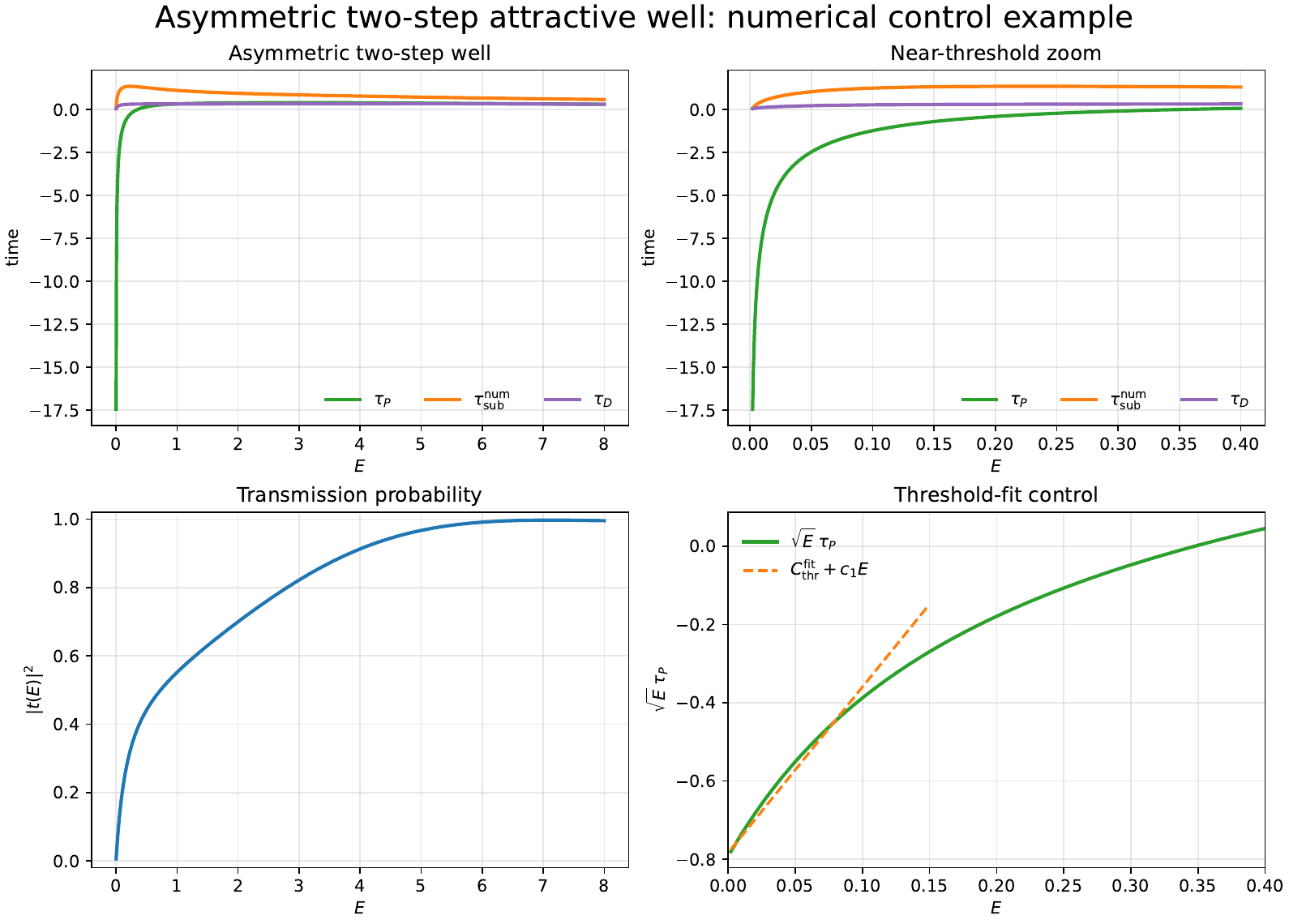}
    \caption{Asymmetric two-step attractive well used as a numerical control example. The raw stationary clock $\tau_P$ exhibits a strong threshold background, while the fitted subtraction produces a much milder $\tau_{\rm sub}^{\rm num}$. The remaining panels show the near-threshold zoom, the transmission probability, and the low-energy fit used to extract $C_{\rm thr}^{\rm fit}$.}
    \label{fig:two-step-control}
\end{figure}

\FloatBarrier

\section{Discussion}
\label{sec:discussion}

The central lesson from this work is that the raw stationary SWP/Peres clock mixes two effects that should be separated: a continuum-edge contribution fixed by low-energy scattering data and the pole-sensitive resonant response.  In the attractive square well this separation is analytic, and with our current normalization, the local dwell-time peak and the threshold-subtracted clock peak have the same leading height near an isolated resonance. The dwell time therefore already sees the resonant structure without the generic threshold obstruction, while the raw stationary clock and the transmission Wigner phase delay both inherit that obstruction. 

The barrier--well--barrier cavity and the numerical asymmetric two-step attractive well together support the interpretation that the subtraction mechanism is not simply a square-well phenomenon. The cavity shows analytically that, once the smooth threshold contribution is separated off, the remaining stationary clock response has the expected local Breit--Wigner form near an isolated resonance, while the two-step well provides a complementary numerical non-square control example in which the threshold background and its subtraction remain stable.

A natural question is whether the same mechanism extends beyond the present one-dimensional setting. The underlying threshold logic suggests that an analogous effect should arise at least in the $s$-wave sector of short-range three-dimensional scattering. Indeed, for low-energy three-dimensional scattering one has $\delta_0(k)\sim -a_s k$ in the $s$-wave channel, with $a_s$ the $s$-wave scattering length, so any stationary clock quantity built from the corresponding phase derivative is expected to inherit an analogous $1/\sqrt{E}$ threshold background, whereas higher partial waves are less singular at threshold~\cite{Wigner1955,Kelkar2007}. In that sense, the threshold subtraction developed here may have a broader interpretation as a way of separating universal continuum-edge kinematics from pole-sensitive delay, although a systematic three-dimensional formulation lies beyond the scope of the present work.

\acknowledgments
We are grateful to Frank Wilczek for useful discussions. DAE is supported by the U.S. Department of Energy,
Office of High Energy Physics, under Award Number DE-SC0019470.

\appendix

\section{Low-energy expansion for compactly supported potentials}
\label{app:low-energy}

In this appendix we derive Theorem~\ref{thm:general-threshold-body} for a real compactly supported potential. The only input is the standard transfer-matrix description of one-dimensional scattering together with the even-in-$k$ structure of the fundamental solutions.

\subsection{Transfer matrix representation}

Consider
\begin{equation}
-\psi''(x)+V(x)\psi(x)=k^2\psi(x),
\qquad k>0,
\label{eq:app-sch}
\end{equation}
with $V$ real and supported in $[0,L]$.

Let $u(x,k)$ and $v(x,k)$ be the fundamental solutions on $[0,L]$ determined by
\begin{equation}
u(0,k)=1,
\qquad u'(0,k)=0,
\qquad
v(0,k)=0,
\qquad v'(0,k)=1.
\label{eq:uv-initial}
\end{equation}
The transfer matrix across the support is
\begin{equation}
\mathsf T(k):=
\begin{pmatrix}
u(L,k) & v(L,k)\\
u'(L,k) & v'(L,k)
\end{pmatrix}
=
\begin{pmatrix}
A(k) & B(k)\\
C(k) & D(k)
\end{pmatrix}.
\label{eq:T-def}
\end{equation}
Equivalently, \(\mathsf T(k)\) maps the initial data
\((\psi(0),\psi'(0))\) to \((\psi(L),\psi'(L))\). It is therefore a
Cauchy-data, or initial-value, transfer matrix rather than an
incoming/outgoing amplitude transfer matrix.
Because the Wronskian of $u$ and $v$ is unity,
\begin{equation}
A(k)D(k)-B(k)C(k)=1.
\label{eq:det-one}
\end{equation}

For left incidence,
\begin{equation}
\psi(x)=
\begin{cases}
e^{ikx}+r(k)e^{-ikx}, & x<0,\\
t(k)e^{ikx}, & x>L.
\end{cases}
\label{eq:scatt-ansatz-app}
\end{equation}
Matching the Cauchy data at $x=0$ and $x=L$ yields
\begin{equation}
t(k)e^{ikL}
=
\frac{2ik}{\Delta(k)},
\qquad
\Delta(k):=
ik\big(A(k)+D(k)\big)+k^2 B(k)-C(k).
\label{eq:t-transfer}
\end{equation}

\subsection{Evenness and the zero-energy criterion}

Since the equation depends on $k$ only through $k^2$, the transfer-matrix entries are even functions of $k$. Hence, as $k\to0$,
\begin{align}
A(k)&=A_0+A_2k^2+A_4k^4+O(k^6), &
B(k)&=B_0+B_2k^2+B_4k^4+O(k^6), \nonumber\\
C(k)&=C_0+C_2k^2+C_4k^4+O(k^6), &
D(k)&=D_0+D_2k^2+D_4k^4+O(k^6),
\label{eq:abcd-even}
\end{align}
with all coefficients real because $V$ is real.

The generic/exceptional distinction is encoded in the zero-energy problem
\begin{equation}
-u''(x)+V(x)u(x)=0.
\label{eq:zero-energy-eq-app}
\end{equation}
Let $u(x,0)$ denote the solution normalized by $u(0,0)=1$, $u'(0,0)=0$. For $x<0$ this extends as the constant solution $u\equiv1$, while for $x>L$ one has
\begin{equation}
u(x,0)=A_0+C_0(x-L).
\label{eq:u-right-zero}
\end{equation}
Therefore $u(x,0)$ is bounded on $\mathbb R$ if and only if
\begin{equation}
C_0=0.
\label{eq:C0-criterion}
\end{equation}
Thus \(C_0=0\) is the zero-energy-resonant, or half-bound, case:
a solution that is constant on the left remains bounded rather than growing
linearly on the right. Accordingly:
\begin{itemize}
\item the \emph{generic} sector is $C_0\neq0$;
\item the \emph{exceptional} sector is $C_0=0$.
\end{itemize}

\subsection{Generic sector}

Assume first that $C_0\neq0$. In this case the denominator
\(\Delta(k)\) has a nonzero real constant term, so the leading threshold
phase is controlled by the ratio of the \(O(k)\) imaginary term to this
constant term. Introduce the real combinations
\begin{equation}
S_0:=A_0+D_0,
\qquad
S_2:=A_2+D_2,
\qquad
R_2:=B_0-C_2.
\label{eq:generic-combos}
\end{equation}
Using \eqref{eq:abcd-even} in \eqref{eq:t-transfer} gives
\begin{equation}
\Delta(k)
=
-C_0+iS_0k+R_2k^2+iS_2k^3+O(k^4).
\label{eq:Delta-generic}
\end{equation}
Factor out the constant term:
\begin{equation}
\Delta(k)
=
-C_0\Bigl[1-ia\,k-b\,k^2-ic\,k^3+O(k^4)\Bigr],
\label{eq:Delta-generic-factor}
\end{equation}
where
\begin{equation}
a:=\frac{S_0}{C_0},
\qquad
b:=\frac{R_2}{C_0},
\qquad
c:=\frac{S_2}{C_0}.
\label{eq:abc-generic}
\end{equation}
Expanding the inverse,
\begin{align}
\frac{1}{\Delta(k)}
&=
-\frac{1}{C_0}
\frac{1}{1-ia\,k-b\,k^2-ic\,k^3+O(k^4)}
\nonumber\\
&=
-\frac{1}{C_0}
\Bigl[
1+ia\,k+(b-a^2)k^2+i(c+2ab-a^3)k^3+O(k^4)
\Bigr].
\label{eq:generic-inverse}
\end{align}
Hence
\begin{equation}
t(k)e^{ikL}
=
-\frac{2ik}{C_0}
\Bigl[
1+ia\,k+(b-a^2)k^2+i(c+2ab-a^3)k^3+O(k^4)
\Bigr].
\end{equation}
Define
\begin{equation}
\mathcal A:=-\frac{2}{C_0},
\qquad
\ell_{\rm thr}:=-\frac{S_0}{2C_0}=-\frac{A_0+D_0}{2C_0},
\label{eq:A-ell-def}
\end{equation}
and also
\begin{equation}
\beta:=b-a^2,
\qquad
\gamma:=c+2ab-a^3.
\end{equation}
Then
\begin{equation}
t(k)e^{ikL}
=
i\,\mathcal A\,k
\Bigl(
1-2i\,\ell_{\rm thr}\,k+\beta k^2+i\gamma k^3+O(k^4)
\Bigr),
\label{eq:t-generic-final-app}
\end{equation}
with $\beta,\gamma\in\mathbb R$.

Now write
\begin{equation}
z(k):=
1-2i\,\ell_{\rm thr}\,k+\beta k^2+i\gamma k^3+O(k^4).
\end{equation}
Since
\begin{equation}
\Re z(k)=1+\beta k^2+O(k^4),
\qquad
\Im z(k)=-2\ell_{\rm thr}k+\gamma k^3+O(k^4),
\end{equation}
one has, on any continuous phase branch,
\begin{equation}
\arg z(k)
=
\arctan\!\left(
\frac{-2\ell_{\rm thr}k+O(k^3)}{1+O(k^2)}
\right)
=
-2\ell_{\rm thr}k+O(k^3).
\label{eq:argz-generic}
\end{equation}
Therefore
\begin{equation}
\delta_P(E)
=
\delta_* -2\ell_{\rm thr}k+O(k^3),
\label{eq:delta-generic-app}
\end{equation}
where $\delta_*$ is an $E$-independent constant phase coming from the prefactor $i\mathcal A k$, and hence
\begin{equation}
\tau_P(E)=\dv{\delta_P}{E}
=
-\frac{\ell_{\rm thr}}{\sqrt E}+O(\sqrt E).
\label{eq:tau-generic-app}
\end{equation}

\subsection{Exceptional sector}

Now assume $C_0=0$. Evaluating \eqref{eq:det-one} at $k=0$ gives
\begin{equation}
A_0D_0=1.
\label{eq:AD-one}
\end{equation}
Since $A_0$ and $D_0$ are real, this implies
\begin{equation}
A_0+D_0\neq0.
\label{eq:S0-nonzero}
\end{equation}
Introduce the real combinations
\begin{equation}
S_0:=A_0+D_0,
\qquad
S_2:=A_2+D_2,
\qquad
R_2:=B_0-C_2,
\qquad
R_4:=B_2-C_4.
\label{eq:exc-combos}
\end{equation}
Then \eqref{eq:t-transfer} gives
\begin{equation}
\Delta(k)
=
iS_0k+R_2k^2+iS_2k^3+R_4k^4+O(k^5).
\label{eq:Delta-exc}
\end{equation}
Factor out the linear term:
\begin{equation}
\Delta(k)
=
iS_0k
\Bigl[
1-ia\,k+b\,k^2-ic\,k^3+O(k^4)
\Bigr],
\label{eq:Delta-exc-factor}
\end{equation}
where
\begin{equation}
a:=\frac{R_2}{S_0},
\qquad
b:=\frac{S_2}{S_0},
\qquad
c:=\frac{R_4}{S_0}.
\label{eq:abc-exc}
\end{equation}
Expanding the inverse,
\begin{align}
\frac{1}{\Delta(k)}
&=
\frac{1}{iS_0k}
\frac{1}{1-ia\,k+b\,k^2-ic\,k^3+O(k^4)}
\nonumber\\
&=
\frac{1}{iS_0k}
\Bigl[
1+ia\,k-(a^2+b)k^2+i(c-2ab-a^3)k^3+O(k^4)
\Bigr].
\label{eq:exc-inverse}
\end{align}
Hence
\begin{equation}
t(k)e^{ikL}
=
\frac{2}{S_0}
\Bigl[
1+ia\,k-(a^2+b)k^2+i(c-2ab-a^3)k^3+O(k^4)
\Bigr].
\end{equation}
Define
\begin{equation}
t_0:=\frac{2}{S_0}=\frac{2}{A_0+D_0},
\qquad
\tilde\ell_{\rm thr}:=-\frac{R_2}{2S_0}
=-\frac{B_0-C_2}{2(A_0+D_0)},
\label{eq:t0-elltilde-def}
\end{equation}
together with
\begin{equation}
\tilde\beta:=-(a^2+b),
\qquad
\tilde\gamma:=c-2ab-a^3.
\end{equation}
Then
\begin{equation}
t(k)e^{ikL}
=
t_0
\Bigl(
1-2i\,\tilde\ell_{\rm thr}\,k+\tilde\beta k^2+i\tilde\gamma k^3+O(k^4)
\Bigr),
\label{eq:t-exc-final-app}
\end{equation}
with $t_0\in\mathbb R\setminus\{0\}$ and $\tilde\beta,\tilde\gamma\in\mathbb R$.

Exactly as in the generic sector,
\begin{equation}
\arg\!\Bigl(
1-2i\,\tilde\ell_{\rm thr}\,k+\tilde\beta k^2+i\tilde\gamma k^3+O(k^4)
\Bigr)
=
-2\tilde\ell_{\rm thr}k+O(k^3),
\end{equation}
so
\begin{equation}
\delta_P(E)
=
\arg t_0 -2\tilde\ell_{\rm thr}k+O(k^3),
\label{eq:delta-exc-app}
\end{equation}
and therefore
\begin{equation}
\tau_P(E)
=
-\frac{\tilde\ell_{\rm thr}}{\sqrt E}+O(\sqrt E).
\label{eq:tau-exc-app}
\end{equation}

\subsection{Conclusion}

Equations \eqref{eq:tau-generic-app} and \eqref{eq:tau-exc-app} establish the generic and exceptional $1/\sqrt E$ threshold laws for the stationary Peres clock and justify the subtraction
\begin{equation}
\tau_{\rm sub}(E):=
\tau_P(E)+\frac{\ell_{\rm thr}}{\sqrt E}
\label{eq:tau-sub-app}
\end{equation}
in the generic sector. The square-well formulas in the main text are obtained by evaluating the corresponding transfer-matrix coefficients explicitly for that model.

\begin{remark}
The compact-support assumption is adopted only to keep the argument concise. The same structure extends to standard short-range classes for which the low-energy Jost or transfer-matrix expansions exist with the same even/odd parity structure in $k$.
\end{remark}

\bibliography{threshold_clock_refs}

\end{document}